\documentclass[aps,showpacs,prx,superscriptaddress,preprintnumbers,amsmath,amssymb,twocolumn]{revtex4-2}
\usepackage[utf8]{inputenc}
\usepackage[T1]{fontenc}
\usepackage{graphicx}
\usepackage{dcolumn}
\usepackage{appendix}
\usepackage{amsmath}
\usepackage{bm}
\usepackage{amsthm}
\newtheorem{theorem}{Theorem}
\usepackage{soul}
\usepackage[dvipsnames]{xcolor}
\usepackage{booktabs}
\usepackage{bbold}
\usepackage{mathtools}
\usepackage{float}
\usepackage{dsfont}
\usepackage{siunitx}
\usepackage[normalem]{ulem}
\usepackage{natbib}
\usepackage{mathtools}
\usepackage{tikz}
\usepackage{xcolor}
\usetikzlibrary{decorations.pathreplacing,calligraphy}

\newtheorem{lemma}[theorem]{Lemma}
\newtheorem{prop}[theorem]{Proposition}

\usepackage{caption}

\def\proj#1{\ket{#1}\!\bra{#1}}

\def\CC{{\mathbb C}}

\usepackage{subfig}
\usepackage{bbm}
\usepackage{enumerate}

\usepackage{hyperref}
\hypersetup{
    colorlinks=true,
    linkcolor=blue,
    filecolor=magenta,      
    urlcolor=cyan
    }

\usepackage{titlesec}

\usepackage{array}

\usepackage{physics}

\usepackage[final]{changes}

\begin{document}
\allowdisplaybreaks
\title{A semi-definite programming formulation of the device-dependent guessing probability}

\author{Raffaele D'Avino}
\affiliation{ICFO-Institut de Ci\`encies Fot\`oniques, The Barcelona Institute of Science and Technology,
Av. Carl Friedrich Gauss 3, 08860 Castelldefels (Barcelona), Spain}

\author{Aurora Mugnai}
\affiliation{ICFO-Institut de Ci\`encies Fot\`oniques, The Barcelona Institute of Science and Technology,
Av. Carl Friedrich Gauss 3, 08860 Castelldefels (Barcelona), Spain}

\author{Miguel Navascués}
\affiliation{Institute for Quantum Optics and Quantum Information (IQOQI) Vienna,\\
 Austrian Academy of Sciences, Boltzmanngasse 3, Wien 1090, Austria}

\author{Antonio Acín}
\affiliation{ICFO-Institut de Ci\`encies Fot\`oniques, The Barcelona Institute of Science and Technology,
Av. Carl Friedrich Gauss 3, 08860 Castelldefels (Barcelona), Spain}
\affiliation{ICREA - Instituci\'o Catalana de Recerca i Estudis Avan\c{c}ats, 08010 Barcelona, Spain}

\author{Gabriel Senno}
\affiliation{Quside Technologies S.L., C/Esteve Terradas 1, 08860 Castelldefels,
Barcelona, Spain}
\affiliation{ICFO-Institut de Ci\`encies Fot\`oniques, The Barcelona Institute of Science and Technology,
Av. Carl Friedrich Gauss 3, 08860 Castelldefels (Barcelona), Spain}

\begin{abstract}
In quantum mechanics, a measurement applied to a state in general produces some amount of intrinsic randomness. This is not only a fundamental feature of the theory, but is also at the basis of any quantum process to generate random numbers. The simplest of such processes consists of a single, fully charaterized, measurement acting on a single, fully characterized, state. Unfortunately, no general method to estimate the intrinsic randomness produced in such setups is known.  In this work, we address this issue by presenting a semidefinite programming formulation of the maximum probability with which an adversary, Eve, can guess the outcomes of characterized but untrusted prepare-and-measure setups. We then present several applications of this construction. First, we apply our method to a variety of specific setups, allowing us both to benchmark the approach and, more importantly, to determine the exact amount of certifiable randomness in scenarios where only upper bounds were previously available. Then, we show that the presence of entanglement between the device preparing the state and the measurement strictly increases Eve's predictive power, already in the most elementary setup of a binary measurement acting on a qubit state. 
\end{abstract}

\maketitle

\section{Introduction}

Randomness is a fundamental resource in both classical and quantum information processing, underpinning applications ranging from cryptography and numerical simulation to foundational tests of physical theories. Quantum mechanics provides a principled source of randomness: when a measurement is performed on a quantum system prepared in a state that is not an eigenstate of the measured observable, the outcome cannot be predicted with certainty, even in principle. This intrinsic unpredictability distinguishes quantum randomness from classical stochasticity and forms the basis of quantum random number generators (QRNG)~\cite{herrero2017quantum,bera2017randomness,mannalatha2023comprehensive}.

An operational way of quantifying the amount of intrinsic randomness produced in a quantum setup is via the maximum probability with which an adversary, Eve, can guess its outcomes given side information about it. In this work, we focus on the simplest of such processes, in which a system is prepared in a known quantum state and a known measurement is applied to it, defining a fully characterized, or device-dependent, prepare-and-measure (PM) setup. In the most general and realistic situation, both the prepared state and the measurement are noisy and, therefore, Eve is allowed to have side information about both of them. In~\cite{senno2023quantifying}, a general definition for the guessing probability in this scenario was provided. Unfortunately, the resulting definition a priori involves a nonconvex optimization for which no general computational procedure is known.

In this work, leveraging the techniques from \cite{navascues2007bounding}, we prove that the device-dependent guessing probability \cite{senno2023quantifying} can be expressed as a semidefinite program (SDP), establishing a computationally-friendly method to estimate the randomness produced in device-dependent PM setups. This construction is not only relevant from a fundamental perspective, but also for the certification of the simplest QRNG. To illustrate the general applicability of this new computational tool, we apply it to three representative situations. First, we consider a simple scenario in which a system defined by a pure qubit state mixed with depolarizing noise is measured along a complementary basis also mixed with depolarizing noise. We show that the analytic upper bound derived in \cite{curran2025maximal} on the certifiable randomness is tight. Next, we improve the analysis of the amount of intrinsic randomness generated in a scheme targeting an implementation in quantum computers put forward by Berta and Brandão in~\cite{berta2021robust}, tightening the bounds provided therein. And, third, we consider the source-device-independent randomness generation scheme introduced in~\cite{avesani2022unbounded}. Therein, the authors provide a family of $k$-outcome measurements with which, if trusted, unbounded randomness can be extracted from an uncharacterized and untrusted qubit system. 
We use our SDP formulation to show that if Eve is allowed to be correlated with the characterized measurements, the amount of intrinsic randomness already vanishes for $k=4$. As the last result of this manuscript, we show that the presence of entanglement between the preparation and the measurement devices leads to a strictly higher guessing probability than if only classical correlations are allowed, already for a scenario consisting of a binary measurement on a qubit. To our knowledge, this is the simplest example of the consequences of entanglement assistance on randomness certification in PM setups~\cite{d2025entanglement,carceller2025role}.


The remainder of the manuscript is organized as follows. In Section~\ref{sec:prelim} we introduce the PM scenario with characterized but untrusted devices and the guessing probability from~\cite{senno2023quantifying}. In Section~\ref{sec:hierarchy}, we present the SDP formulation used to compute this quantity. Section~\ref{sec:application} illustrates the method through three applications. Finally, in Section~\ref{sec:separability} we use our formulation to show that entanglement between the preparation and measurement devices can enhance Eve’s predictive power.


\section{The device-dependent guessing probability}\label{sec:prelim}

\begin{figure}
\begin{tikzpicture}[scale=1]
    %
	\node[] at (3.5,0.7) {$E$};
	\draw[] (3.2,0.3) rectangle (3.8,1);

  \draw[] (-3.2,1) rectangle (3,-1.8);
  \node[] at (-2.4,-1.6) {Alice's lab};  
   \draw[rounded corners=10pt] (-0.95,-0.75) rectangle (-2.95,0.75);
    \node[] at (-1.95,0) {{\huge $\mathrm{P}$}};

    \draw[dashed,rounded corners=10pt] (2.8,-0.75) rectangle (0.55,0.75);
    \node[] at (0.9,0.4) {{\large $\mathrm{M}$}};
    \draw[rounded corners=5pt] (1.4,-0.5) rectangle (2.7,0.3);
    {\fontsize{0.3cm}{0.5em}\selectfont \node[] at (2.05,-0.1) {$\{\Pi^a_{SM}\}_a$};}
   
	\draw[-,line width=0.05cm] (-0.95,0) to (0.5,0);
\draw[->,line width=0.05cm,dashed] (0.55,0) to (1.4,0);

   {\fontsize{0.3cm}{0.5em}\selectfont\node[] at (0.8,-0.4) {$\rho_M$};}
   
\draw[->,line width=0.05cm,dashed] (1,-0.4) to (1.4,-0.4);

	\node[] at (-0.15,0.2) {{\color{orange}\large $\rho_S$}};

	\draw[-,line width=0.05cm,dashed] (1.8,-0.5) to (1.8,-0.7);
	\draw[->,line width=0.05cm] (1.8,-0.75) to (1.8,-1.3);
    {\fontsize{0.35cm}{0.5em}\selectfont  \node[] at (1.8,-1.5) {$A\sim \Tr[{\color{orange}M^a_{S}}\rho_{S}]$};}

    \draw [decorate,
    decoration = {brace,mirror},line width=0.05cm] (3.9,-1.8) --  (3.9,1);
    \node[] at (4.6,-0.5) {$\ket{\psi}_{SME}$};

  \end{tikzpicture}
  \caption{\textbf{Characterized but untrusted PM setup.} Alice holds a PM device that prepares a system $S$ in some {\color{orange}known state} $\rho_S$. In the measurement box $\mathrm{M}$, an unknown PVM $\{\Pi_{SM}^a\}_a$ is performed, whose outcome statistics are described by the {\color{orange}known POVM} $\{M_S^a\}_a$. Eve holds a purification $\ket{\psi}_{SME}$ of the joint state $\rho_{SM}$ of the system $S$ and the ancillary system $\rho_M$ inside $M$.}\label{fig:scenario} 
\end{figure}
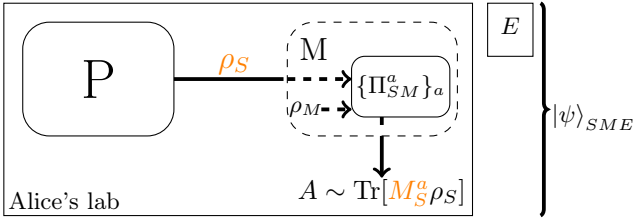
We consider a scenario in which an honest user, Alice, holds a device that can prepare a $d$-dimensional quantum system $S$ in some known state $\rho_S$ and then performs a $k$-outcome measurement on it, described by a known positive operator-valued measure (POVM) $M_S=\{M_S^a\}_{1\leq a \leq k}$. The goal is to determine how random the measurement's outcomes are from the point of view of an adversary, Eve, who holds quantum side-information $E$ about Alice's device. To model Eve’s side information, we follow the approach of \cite{senno2023quantifying} and consider that what Alice \emph{describes} as the POVM $M_S$ is, in fact, a projective measurement (PVM) $\{\Pi_{SM}^a\}_a$ on the system $S$ plus an ancillary system $M$. The state $\rho_M$ of the ancilla together with the PVM $\{\Pi_{SM}^a\}_a$ form a so-called generalized Naimark dilation of $M_S$ \cite{frauchiger2013true}. See Fig.~\ref{fig:scenario} for a schematic description of the scenario.

We assume that Alice only has tomographic access to the system $S$. This implies that all the infinitely many dilations compatible with $M_S$ are indistinguishable for her. Moreover, she can neither distinguish whether system $S$ is correlated with system $M$ or not. Therefore, we assume that the dilation has been chosen by Eve, who also holds a purification $\ket{\psi}_{SME}$ of the state $\rho_{SM}$ of the system+ancilla. Eve may then perform an arbitrary measurement $\{M_E^a\}_a$ on her system to guess the outcome of the measurement on $S$. This leads to the following definition \cite{senno2023quantifying} of \emph{Eve's guessing probability} 
\begin{equation}\label{eq:Pguess}
    \begin{aligned}
        &P_{\text{guess}}(A\mid E,\rho_S,\{M_S^a\}_a):=&&\\ 
        &\quad\max_{\{\Pi_{SM}^{a}\}_{a},\{M_E^e\}_e,\ket{\psi}_{SME}}\sum_{a}\bra{\psi}\Pi_{SM}^{a}\otimes M_E^a \ket{\psi}_{SME}&&\\
        &~\textrm{subject to } &\\
        &\quad\Tr_{ME}[\ket{\psi}\bra{\psi}_{SME}]=\rho_{S},&&\\
        &\quad\Tr_{M}[\Pi_{SM}^{a}(\mathbb{1}_{S}\otimes\Tr_{SE}{[\ket{\psi}\bra{\psi}_{SME}]})]=M_{S}^a\quad \forall a,&&\\
        &\quad\bra{\psi}\Pi_{SM}^{a}\otimes\mathbb{1}_{E}\ket{\psi}=\Tr[M_{S}^{a}\rho_{S}]\quad \forall a.&&
    \end{aligned}
\end{equation}
The constraints above enforce that Eve's attack is compatible with the fully characterized state and measurement, and the observed statistics.

The quantity of interest from a randomness certification perspective is the conditional min-entropy,
\begin{align}\label{eq:min-entropy}
&H_{\min}(A|E,\rho_S,\{M_S^a\}_a) :=&\nonumber\\
&\qquad\qquad\qquad\qquad-\log_2 P_{\text{guess}}(A|E,\rho_S,\{M_S^a\}_a).
\end{align}
which measures the amount of randomness in the measurement outcome $A$, from the perspective of an adversary holding quantum side information $E$ \cite{konig2009operational}.

\medskip

Eq.~\eqref{eq:Pguess} can be reformulated by introducing subnormalized states of the system+ancilla conditioned on Eve's measurement outcomes,
\begin{equation}
\tilde{\sigma}_{SM}^{a} := \Tr_E\big[(\mathbb{1}_{SM} \otimes M_E^a) \ket{\psi}\bra{\psi}_{SME}\big].
\end{equation}
The guessing probability can then be expressed as
\begin{equation}\label{eq:pguess_SM}
\begin{aligned}
&P_{\text{guess}}(A\mid E,\rho_S,\{M_S^a\}_a) = &&\\
&\quad \max_{\{\Pi_{SM}^{a}\}_{a}, \{\tilde{\sigma}_{SM}^{e}\}_e} \sum_a \Tr[\Pi_{SM}^{a} \tilde{\sigma}_{SM}^{a}] &&\\
&\text{subject to} &&\\
&\quad \Tr_M\Big[\sum_e \tilde{\sigma}_{SM}^e\Big] = \rho_S, &&\\
&\quad \Tr_M\Big[\Pi_{SM}^a (\mathbb{1}_S \otimes \Tr_S[\sum_e \tilde{\sigma}_{SM}^e])\Big] = M_S^a, \quad \forall a, &&\\
&\quad \Tr[\Pi_{SM}^a \sum_e \tilde{\sigma}_{SM}^e] = \Tr[M_S^a \rho_S], \quad \forall a. &&
\end{aligned}
\end{equation}

While this equivalent representation remains a noncommutative polynomial optimization problem, it provides a convenient starting point for the SDP construction presented in the next section.

\section{SDP formulation}\label{sec:hierarchy}
In this section, we describe our main contribution: an SDP formulation for Eq.~\eqref{eq:pguess_SM}.
The simple but key observation underlying our construction is that the projective measurement $\{\Pi_{SM}^a\}_a$ in any dilation of $\{M_S^a\}_a$ can be expressed in block form with respect to an orthonormal basis $\{\ket{i}\}_{i=1}^{d}$ of $\mathcal{H}_S$ as
\begin{equation}\label{eq:proj_meas}
    \Pi_{SM}^a = \sum_{i,j=1}^{d} \ket{i}\!\bra{j}_S \otimes K_M^{i,j|a}
    \qquad \forall a ,
\end{equation}
where $d={\rm dim}(\mathcal{H}_S)$ and $\{K_M^{i,j|a}\}_{i,j,a}$ is a set of operators acting on the auxiliary system (of unrestricted dimension) $M$. 
This observation allows us to apply the techniques 
from \cite{navascues2014characterization} and consider 
the \emph{generalized moment matrices}%
\addtocounter{footnote}{1} \footnotemark[\value{footnote}]
\begin{equation}\label{eq:mom_matrix}
    \Gamma_{\text{e}}:=\sum_{s, t\in \cal O}\sum_{i,j=1}^{d}
    c^{i,j}_e(s^\dagger t)\,\ket{i}\bra{j}\otimes\ket{s}\bra{t},
\end{equation}
\footnotetext[\value{footnote}]{We abuse notation by using 
$s$ and $t$ to denote both the operators in $\mathcal{O}$ 
and the labels indexing the corresponding basis vectors 
$\ket{s}$, $\ket{t}$ of the auxiliary Hilbert space.}
with 
\begin{align}\label{eq:coeff_gamma}
c^{i,j}_e(s^\dagger t)=\Tr\left[(\ket{j}\bra{i}\otimes s^\dagger t)  \tilde{\sigma}^e_{SM}\right].
\end{align}
Here, the outer sum runs over the finite set of operators
\begin{equation}\label{eq:monomial_set}
{\cal O}:=\{\mathbb{1}\}\cup \{K^{i,j|a}_M:a=1,...,k;\, i,j=1,...,d\}.
\end{equation}
As described in \cite{navascues2014characterization}, the matrices $\Gamma_{\text{e}}$ can be seen as the result of the following mapping
\begin{equation}
    \tilde{\sigma}_{SM}^e\longrightarrow
    \Tr_M\!\left[
        (\mathbb{1}_S \otimes \Lambda)\,
        \tilde{\sigma}_{SM}^e\,
        (\mathbb{1}_S \otimes \Lambda^\dagger)
    \right]
    =:\Gamma_e,
\end{equation}
with
$\Lambda
    =
    \sum_s
    s^\dagger \otimes \ket{s}.$
This, in particular, implies their positive-semidefiniteness, $\Gamma_{\text{e}}\geq 0$ for all $e$.

In Appendix ~\ref{App:mom_matr_formulation}, we show that the objective function and the constraints in Eq. \eqref{eq:pguess_SM} can be expressed as linear functions of the coefficients of the matrices $\{\Gamma_e\}_e$. Adopting the shorthand $(K_M^{i,j|a})^\dagger= K_M^{j,i|a}$ for convenience, this provides the following \emph{semidefinite relaxation}
\begin{equation}\label{eq:pguess_matrices}
\begin{aligned}
&P_{\text{guess}}(A\mid E,\rho_S,\{M_S^a\}_a) = &&\\
&\quad \max_{\{\Gamma_e\}_e}\sum_{a,i,j}
    \bra{j}\!\bra{\mathbb{1}}\,
    \Gamma_a\,
    \ket{i}\!|K_M^{i,j|a}\rangle &&\\
&\text{subject to} &&\\
&\quad \Gamma_e\succeq 0 \;\, \forall e,&&\\
&\quad \sum_e
    \bra{i}\!\bra{\mathbb{1}}\,
    \Gamma_e\,
    \ket{j}\!\ket{\mathbb{1}} = \bra{i}\rho_S\ket{j} \;\, \forall i,j,&&\\
&\quad \sum_{e,l}
    \bra{l}\bra{\mathbb{1}}\,
    \Gamma_e\,
    \ket{l}\!|K_M^{i,j| a}\rangle = \bra{i} M_S^a \ket{j} \;\, \forall a,i,j, &&\\
&\quad \sum_{e,i,j}
    \bra{j}\!\bra{\mathbb{1}}\,
    \Gamma_e\,
    \ket{i}\!|K_M^{i,j|a}\rangle = \Tr[M_S^a \rho_S]\;\, \forall a, &&\\
&\quad \bra{m}\!\bra{\mathbb{1}}\,
    \Gamma_e\,
    \ket{n}\!|s\rangle = \bra{m}\!\bra{s^\dagger}\,
    \Gamma_e\,
    \ket{n}\!|\mathbb{1}\rangle  &&\\&\quad\qquad\qquad\qquad\qquad\qquad\quad\quad\forall s\in\mathcal{O}\quad\forall\, e,m,n, &&\\
&\quad\sum_a\bra{m}\!\bra{\mathbb{1}}\,
    \Gamma_e\,
    \ket{n}\!|K_{M}^{i,j|a}\rangle = \delta_{ij}\bra{m}\!\bra{\mathbb{1}}\,
    \Gamma_e\,
    \ket{n}\!|\mathbb{1}\rangle &&\\&\quad\qquad\qquad\qquad\qquad\qquad\qquad\qquad\;\,\forall\, e,m,n,i,j,&&\\
&\quad
\sum_l \bra{m}\langle (K_M^{l,i|a})^\dagger|\,
    \Gamma_e\,
    \ket{n}\!|K_M^{l,j|b}\rangle =  &&\\
    &\quad\qquad\delta_{ab}\bra{m}\!\bra{\mathbb{1}}\,
    \Gamma_e\,
    \ket{n}\!|K_M^{i,j|a}\rangle\;\,\forall\, a,b,e,m,n,i,j.&&
\end{aligned}
\end{equation}
The second, third and fourth constraints are direct analogs of the ones in Eq. \eqref{eq:pguess_SM}. The last three constraints, on the other hand, follow from the hermiticity, the completeness and projectivity of $\{\Pi_{SM}^a\}_a$ respectively 
(see Appendix~\ref{App:mom_matr_formulation}).

\medskip

The previous discussion implies that from a feasible solution of Eq. \eqref{eq:pguess_SM} it is possible to construct a feasible solution of Eq.~\eqref{eq:pguess_matrices}. Our main result, Theorem \ref{th:SDP_formulation} below, shows that this connection also goes in the opposite direction and, thus, Eq. \eqref{eq:pguess_matrices} provides a \emph{complete characterization} of the device-dependent guessing probability.

\begin{theorem}\label{th:SDP_formulation}
Let $\{\Gamma_e\}_e$ be a feasible solution to Eq.~\eqref{eq:pguess_matrices} achieving some objective value $c$. Then, there exist a finite-dimensional Hilbert space $\mathcal{H}_M$, operators $K_M^{i,j|a}$ and (possibly subnormalized) states $\tilde{\sigma}^e_{SM}$ such that $\langle \{\sigma_{SM}^e\}_e,\{\Pi_{SM}^a\}_a\rangle$, with $\Pi_{SM}^a$ as in Eq. \eqref{eq:proj_meas}, is a feasible solution to Eq.~\eqref{eq:pguess_SM} achieving the same objective value $c$.
\end{theorem}
The proof of this theorem is deferred to Appendix \ref{app:proof_SDP_formulation}.

\section{Applications}\label{sec:application}
In this section, we apply the SDP introduced in Section~\ref{sec:hierarchy} to different setups.

\subsection{Depolarized state and measurement}
We first consider a setting that was studied in Ref.~\cite{curran2025maximal}, where both the state and the measurement are equally affected by depolarizing noise.
More formally, the state and the measurement are assumed to be

\begin{equation}\label{eq:noisy_state_meas}
\begin{aligned}
    \rho_S =\,& \varepsilon \ket{+}\bra{+} + (1-\varepsilon)\mathbb{1}/2, \\
    \{M_S^a\}_a =\,& \{\varepsilon\ket{0}\bra{0} + (1-\varepsilon)\mathbb{1}/2,\;
                   \\
                   &\;\,\varepsilon\ket{1}\bra{1}+ (1-\varepsilon)\mathbb{1}/2\},
\end{aligned}
\end{equation}
where $\varepsilon\in[0,1]$ and $|+\rangle=(|0\rangle+|1\rangle)/\sqrt{2}$. The authors prove the following analytical lower bound to the guessing probability
\begin{equation}\label{eq:pguess-lower-bound-global-noise}
    P_{\text{guess}}(A\mid E,\rho_S,\{M_S^a\}_a) \geq \frac{1}{2}\left(1 + 2\sqrt{\delta(1-\delta)}\right),
\end{equation}
with $\varepsilon = 1 - \sqrt{1-\delta}$.
The relation between the noise parameters $\delta$ and $\varepsilon$ is chosen so that the resulting measurement statistics are identical 
in the two scenarios. 

\medskip

In Fig.~\ref{fig:p_guess_Fionnuala}, we plot the value of the conditional min-entropy in Eq.~\eqref{eq:min-entropy}
provided by the SDP in Eq. \eqref{eq:pguess_matrices} for this setting. As can be seen, it matches the analytic value obtained from Eq. \eqref{eq:pguess-lower-bound-global-noise}, establishing the optimality of the latter.

\begin{figure}[h!]
    \centering
    \includegraphics[width=1\linewidth]{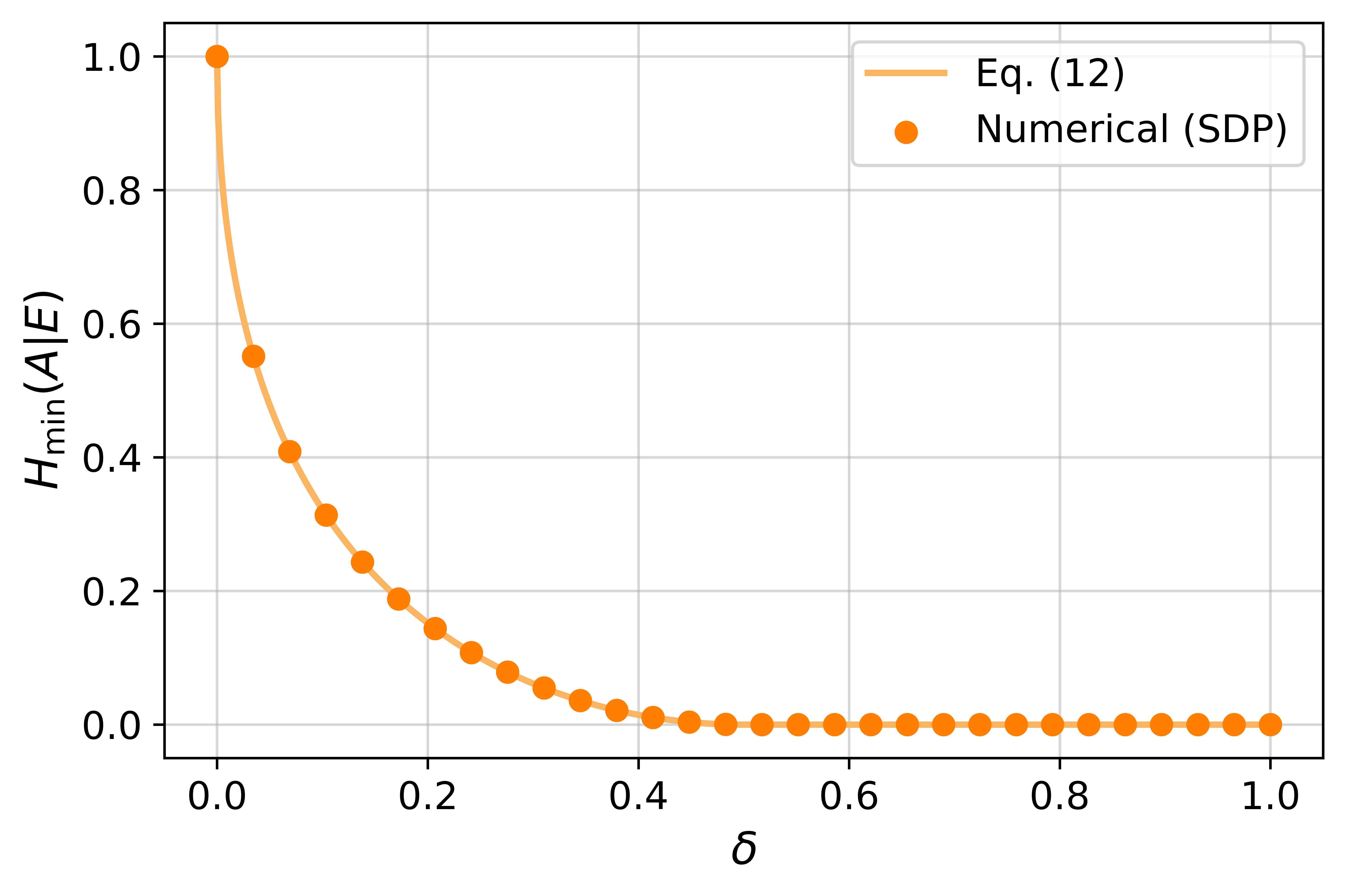}
    \caption{$H_\mathrm{min}(A|E)$ obtained via our SDP for the setting in Eq.~\eqref{eq:noisy_state_meas} (dots), compared with the analytic results for the corresponding quantity (solid line).}
    \label{fig:p_guess_Fionnuala}
\end{figure}

\subsection{Depolarized state and inefficient detector}
In Ref.~\cite{berta2021robust}, the authors put forward a scheme for randomness generation with quantum computers. Specifically, they consider 
a setup preparing the state $\ket{+}$, a uniform superposition of the computational basis states, which is affected by depolarizing noise of strength $\delta$. The system is later subjected to a binary measurement parametrized by $\mu\in(0,1)$, where smaller values of $\mu$ correspond to a stronger bias towards the read-out of the basis state $\ketbra{0}$. More formally, they consider the following state and measurement
\begin{equation}
\begin{aligned}
    \rho_S &= (1-\delta) \ket{+}\bra{+} + \delta\,\mathbb{1}/2, \\
    \{M_S^a\}_a &= \{\mathbb{1}-\mu |1\rangle\langle1|,\mu |1\rangle\langle1|\}.
\end{aligned}
\end{equation}
To estimate the conditional min-entropy, following the prescription in \cite{frauchiger2013true}, the authors fix a specific dilation $\mathcal{D}$ of the measurement $\{M_S^a\}_a$, given by
\begin{equation}\label{eq:dilation_pguess}
\begin{aligned}
    \{\Pi_{SM}^a\}_a &= \{|0\rangle\langle0|\otimes|0\rangle\langle0|,\;
    |0\rangle\langle0|\otimes|1\rangle\langle1| + \mathbb{1}\otimes|1\rangle\langle1|\}, \\
    \rho_M &= (1-\mu)|0\rangle\langle0| + \mu |1\rangle\langle1|.
\end{aligned}
\end{equation}
This is one of the many possible dilations, not necessarily optimal from Eve's point of view. We compare the result obtained using this specific dilation with the unrestricted one using Eq. \eqref{eq:pguess_matrices}. 
In Fig.~\ref{fig:hmin_BB} we plot $H_{\text{min}}(A|E,\mathcal{D}) - H_{\text{min}}(A|E)$ as a function of the parameters $\delta$ and $\mu$. We observe that for many values of the parameters this quantity is positive, implying an overestimation (albeit small) of the amount of intrinsic randomness generated.

\begin{figure}[h!]
    \centering
    \includegraphics[width=0.9\linewidth]{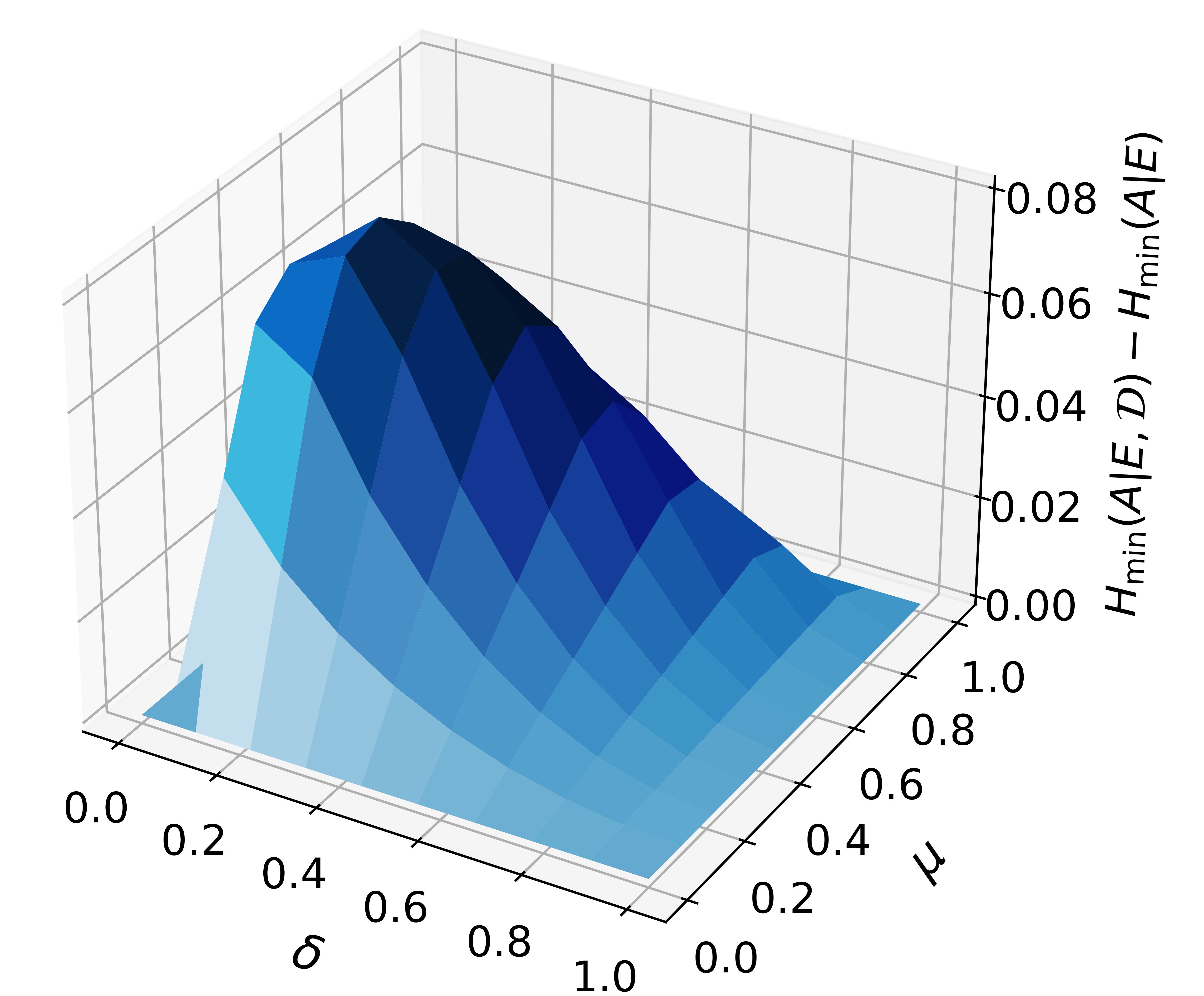}
    \caption{Difference between the conditional min-entropy $H_{\mathrm{min}}(A|E)$ evaluated on the specific dilation in Eq.~\eqref{eq:dilation_pguess} and the value obtained via our SDP optimizing over all possible dilations.}
    \label{fig:hmin_BB}
\end{figure}
\subsection{Uncharacterized state and equiangular planar measurement}


As a final application, we consider the source-device-independent scenario introduced in Ref.~\cite{avesani2022unbounded}. In this framework, no assumption is made on the prepared qubit state $\rho_S$, while the measurement is assumed to be fully characterized and described by an equiangular $k$-outcome POVM $\{M_S^{a}\}_{1\leq a\leq k}$ defined on the $XZ$ plane of the Bloch sphere, with elements
\begin{equation}\label{eq:regular_POVM}
    M_S^{a}=\frac{1}{k}\bigl(\mathbb{1}+\vec{u}_a\cdot\vec{\sigma}\bigr),
\end{equation}
where $\vec{u}_a = (\cos\theta_a,0,\sin\theta_a)$, with $\theta_a=2\pi a/k$, are the vertices of a regular polygon in the $XZ$ plane and $\vec{\sigma}=(\sigma_x,\sigma_y,\sigma_z)$ denotes the vector of Pauli operators. Moreover, the measurement is assumed to be fully trusted, in the sense that the eavesdropper cannot have any side information. This is an important assumption, as any qubit measurement with more than four outcomes is not extremal and, therefore, Eve can be coupled to it in an adversarial model. Ref.~\cite{avesani2022unbounded} assumes that this is not the case. In fact, for measurements restricted to the $XZ$ plane, measurements cease to be extremal when $k>3$.

Under the mentioned assumptions, it is shown in Ref.~\cite{avesani2022unbounded} that the min-entropy diverges, suggesting that unbounded randomness can be certified with qubit systems. In particular, the conditional min-entropy satisfies the lower bound
\begin{equation}
m_n \equiv \min_{\rho_S} H_{\min}(A | E) = \log_2(k) - 1,
\end{equation}
which is represented by the continuous line in Fig.~\ref{fig:Hmin_vs_n}.

To assess whether this conclusion remains valid in a scenario with untrusted measurements, we consider a modified version of the SDP in Eq. \eqref{eq:pguess_SM}. The only difference is that the input state $\rho_S$ is no longer fixed, and is instead included as an optimization variable. Accordingly, the maximization is extended to $\rho_S$, subject to the constraints $\rho_S \succeq 0$ and $\operatorname{Tr}[\rho_S]=1$.

The results reported in Fig.~\ref{fig:Hmin_vs_n} show that, for $k > 3$, assuming a trusted measurement leads to a significant overestimation of the min-entropy. This is a consequence of the lack of extremality of these measurements.

\begin{figure}[h!]
    \centering
    \includegraphics[width=1\linewidth]{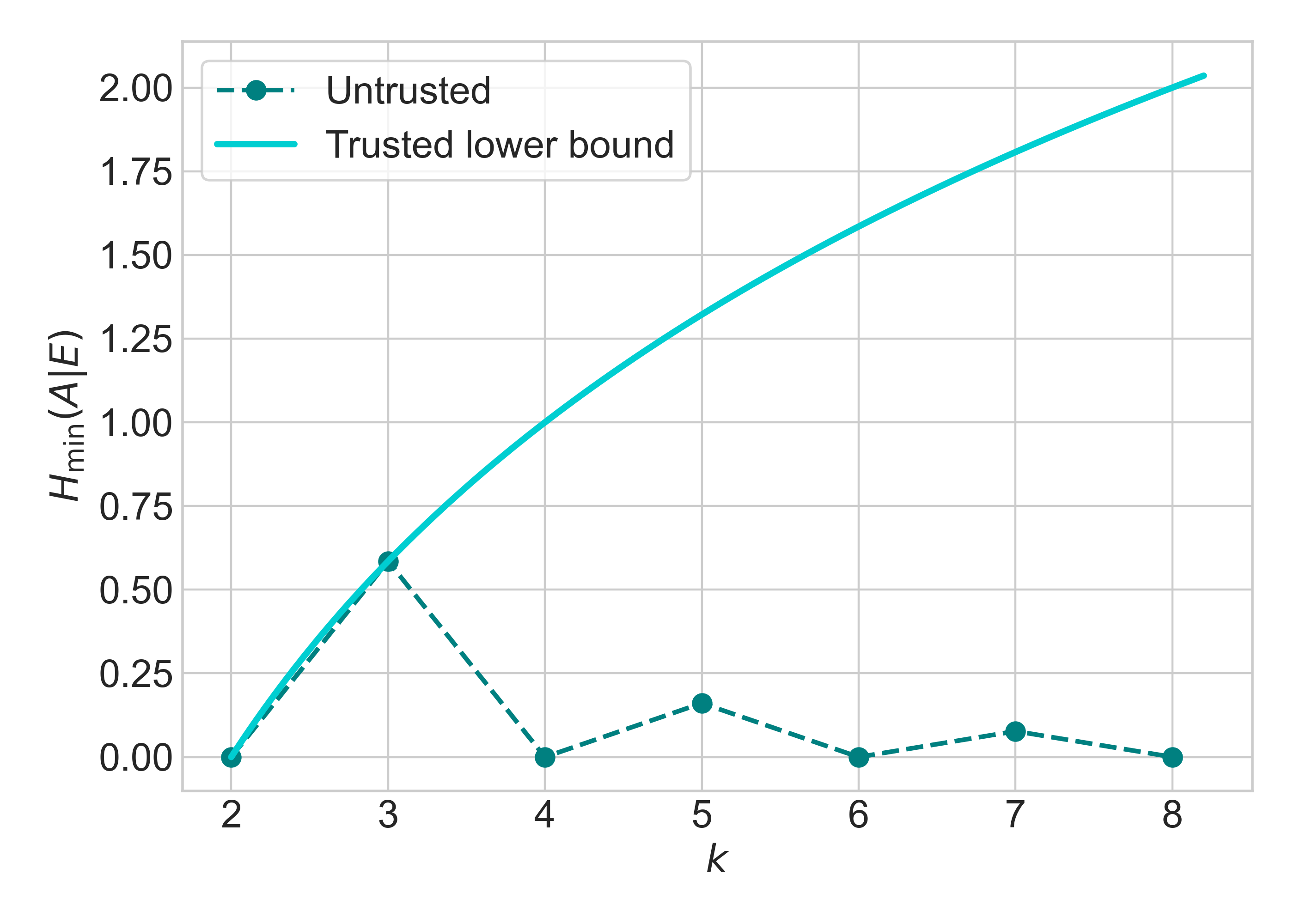}
    \caption{Conditional min-entropy $H_{\min}(A|E)$ as a function of the number of outcomes $k$ for an equiangular POVM in the $XZ$ plane of the Bloch sphere. The solid line represents the theoretical lower bound $m_n = \log_2(n)-1$ assuming fully trusted measurements, while the points show the corresponding values obtained via our SDP formulation, without assuming trusted measurements. For $k > 3$, assuming trusted measurements significantly overestimates the min-entropy.}
    \label{fig:Hmin_vs_n}
\end{figure}

Extending our analysis, we take into account realistic depolarizing noise in the measurement, which further decreases the min-entropy values we observe. We model the measurement by replacing each POVM element $M_S^{a}$ with
\begin{equation}
\tilde{M}_S^{a} = (1-\delta) M_S^{a} + \delta\, \frac{\operatorname{Tr}[M_S^{a}]}{d} \, \mathbb{1}_S,
\end{equation}
where $\delta \in [0,1]$ is the noisy parameter, with $\delta = 1$ corresponding to a fully depolarized measurement.
For $k=3$, for which the measurements is extremal in the absence of noise, Fig.~\ref{fig:Hmin_vs_delta} shows that the min-entropy already vanishes at $\delta \simeq 0.073$.

\begin{figure}[h!]
    \centering
    \includegraphics[width=1\linewidth]{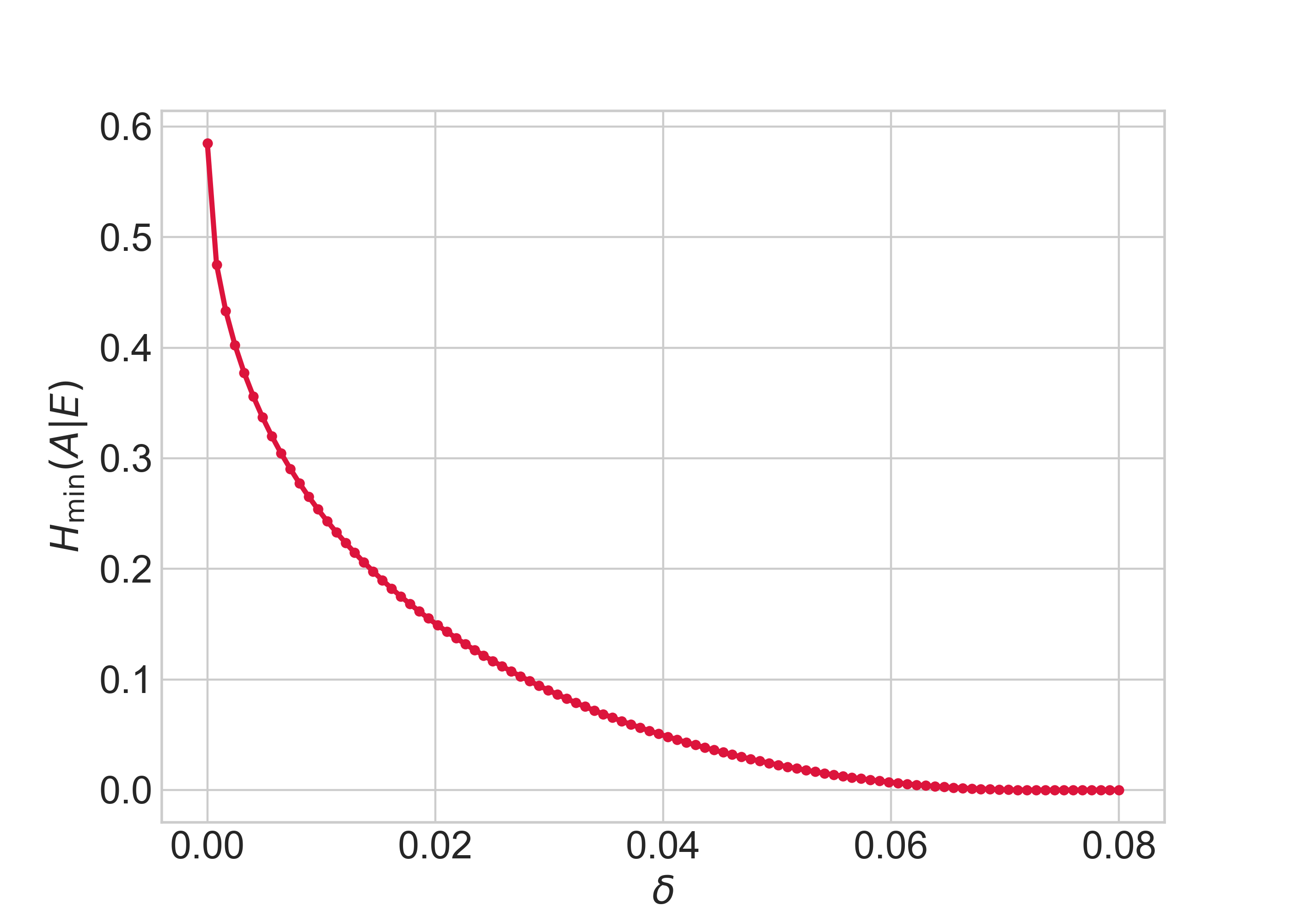}
    \caption{Conditional min-entropy $H_{\min}(A|E)$ for $k=3$ as a function of the depolarizing noise parameter $\delta$ of the measurement. The min-entropy decreases rapidly as the noise increases and vanishes around $\delta \simeq 0.073$.}
    \label{fig:Hmin_vs_delta}
\end{figure}

\section{Entanglement-assistance increases Eve's predicting power}\label{sec:separability}
The global state $\rho_{SM}$ of the measured system $S$ and ancillary system $M$ in Eq.~\eqref{eq:Pguess} is allowed to be correlated. As explained above, the rationale behind this decision follows from the fact that the honest user, Alice, is assumed to perform independent tomographic calibration processes on the prepared system $S$ and implemented measurement $M$, and is thus unable to detect correlations between these two devices. This assumption, however, contrasts with the one performed in the independently introduced randomness quantification framework from \cite{dai2023intrinsic}. While also based on the generalized Naimark dilations from \cite{frauchiger2013true}, the framework in \cite{dai2023intrinsic} deviates from \cite{senno2023quantifying} in assuming that: 1) systems $S$ and $M$ are in a product state $\rho_S\otimes \sigma_M$; and 2) Eve holds systems $E$ and $F$ purifying, respectively, $S$ and $M$, but she cannot measure them jointly. The results from \cite[Thms. 3 and 4]{senno2023quantifying} show that in a setting where 1) holds, assuming 2) incurs in a restriction to Eve's power. Naturally, this leaves open the question of whether 1) alone has any operational consequences. In this section, we answer this question in the affirmative.

\medskip

Let $P_{{\rm guess}}^{\rm sep}$ be as in Eq.~\eqref{eq:Pguess} but with the added constraint of a separable $\rho_{SM}$. In this section, we show that there exist states $\rho_S^*$ and POVMs $M_S^*$ for which
\begin{align*}
P_{{\rm guess}}(A|E,\rho_S^*,M_S^*)>P_{{\rm guess}}^{\rm sep}(A|E,\rho_S^*,M_S^*).
\end{align*}

We make use of the following lemma, whose proof we defer to Appendix~\ref{App:separability}.

\begin{lemma}
Given a state $\rho_{SM}$ that is separable, then
\[
\Gamma^{T_S}(\rho_{SM}) \geq 0 .
\]
\end{lemma}
Here, $T_S$ denotes the partial transposition with respect to subsystem $S$. Therefore, by constraining the sum $\Gamma=\sum_e \Gamma_e$ of the matrices defined in Eq.~\eqref{eq:mom_matrix} to have a positive partial transpose (PPT), we obtain upper bounds to $P_{{\rm guess}}^{\rm sep}$.\\

As a concrete example, for the qubit state $\rho_S$ and the binary measurement $M_S=\{M_S^0,\mathbb{1}_S-M_S^0\}$ given by
\begin{equation}
\rho_S =
\begin{pmatrix}
0.2 & 0.3 \\
0.3 & 0.8
\end{pmatrix},
\qquad
M_S^0 =
\begin{pmatrix}
0.6 & r e^{i\phi} \\
r e^{-i\phi} & 0.1
\end{pmatrix},
\end{equation}
with $r=0.24$, $\phi=0.89$, we find that
\begin{align*}
&P_{{\rm guess}}^{\rm sep} \leq 0.909672 < P_{{\rm guess}} = 0.910376.
\end{align*}
In spite of the gap for this particular example being quite small, it suffices to show that, already in the simplest scenario given by a single qubit and a binary measurement, assuming no entanglement between the system $S$ and the measurement register $M$ can overestimate the amount of intrinsic randomness generated.

\section{Conclusion}

In this work, we studied the amount of intrinsic randomness that can be generated in a characterized but untrusted prepare-and-measure scenario. In particular, we introduced a semidefinite programming formulation of the conditional min-entropy of Eve possessing quantum side information about the experimental setup.

We apply our method in three different scenarios. This allows us to benchmark our formulation and, more importantly, to determine the exact amount of certifiable randomness in settings where previously only upper bounds were known, often relying on additional assumptions to make the conditional min-entropy computable. Moreover, we used our framework to show that restricting the correlations between the preparation and measurement devices to be classical strictly reduces Eve's predictive power.

Our results provide a general tool for quantifying randomness in realistic, noise-affected quantum devices. They not only advance the current understanding of randomness generation in the presence of quantum side information, but also provide a general recipe to certify randomness in the simplest device-dependent QRNG consisting of a known measurement implemented on a known quantum state. 

\section*{ACKNOWLEDGEMENTS}

This work is supported by the Government of Spain (Severo Ochoa CEX2019-000910-S and FUNQIP), Fundació Cellex, Fundació Mir-Puig, Generalitat de Catalunya (CERCA program), the European Union (QSNP, 101114043), and the AXA Chair in Quantum Information Science. 
RD acknowledges funding from the European Union’s Horizon Europe research and innovation programme under the Marie Skłodowska-Curie grant agreement No. 101081441 and thanks Ranieri V. Nery for useful discussions.

\section*{Data and code availability}
The Python code used to perform the numerical optimizations and generate the figures presented in this work is available at \url{https://github.com/raffaeledavino/dd-guessing-probability}.

\bibliography{Reference.bib}

\clearpage
\onecolumngrid
\appendix

\begin{center}
\large{\textbf{\textsc{Appendix}}}

\end{center}

\section{Moment matrix formulation}\label{App:mom_matr_formulation}
We now consider the optimisation problem in Eq.~\eqref{eq:pguess_SM}, which we rewrite here for convenience
\begin{equation}
    P_{\text{guess}}^{Q}=\max_{\{\Pi_{SM}^{a}\}_{a},\{\tilde{\sigma}_{SM}^{e}\}_{e},}\sum_{a=1}^{k}\Tr[\Pi_{SM}^{a}\tilde{\sigma}_{SM}^{a}],
\end{equation}
subject to
\begin{equation}
\begin{aligned}
    &\Tr_{M}\left[\sum_{e=1}^{k}\tilde{\sigma}_{SM}^{e}\right]=\rho_{S},\\
    &\Tr_{M}\left[\Pi_{SM}^{a}\left(\mathbb{1}_{S}\otimes\Tr_{S}{\Big[\sum_{e=1}^{k}\tilde{\sigma}_{SM}^{e}\Big]}\right)\right]=M^a_{S}\quad &\forall a\in\{1,\dots,k\},\\
    &\Tr\left[\Pi_{SM}^{a}\sum_{e=1}^{k}\tilde{\sigma}^{e}_{SM}\right]=\Tr\left[M_{S}^{a}\rho_{S}\right]\quad&\forall a\in\{1,\dots,k\}.
\end{aligned}
\end{equation}
Following the approach in~\cite{navascues2014characterization}, we rewrite this optimisation in terms of moment matrices via the mapping
\begin{equation}
    \tilde{\sigma}_{SM}^e\longrightarrow
    \Tr_M\!\left[
        (\mathbb{1}_S \otimes \Lambda)\,
        \tilde{\sigma}_{SM}^e\,
        (\mathbb{1}_S \otimes \Lambda^\dagger)
    \right]
    :=\Gamma_e,
\end{equation}
where $\Lambda=\sum_{s\in\cal{O}} s^\dag\otimes\ket{s}$, with $\mathcal O:=\{\mathbb{1}\}\cup \{K^{i,j|a}_M:a=1,...,k;\, i,j=1,...,d\}$, and $e\in\{1,\dots,k\}$. Using this definition, the moment matrices $\{\Gamma_{e}\}_\text{e}$ can be written explicitly as
\begin{equation}
\begin{aligned}
        \Gamma_e&=\sum_{s,t\in\mathcal{O}}\Tr_M\left[(\mathbb{1}_S\otimes s^\dagger\otimes\ket{s})\tilde{\sigma}_{SM}^e(\mathbb{1}_S\otimes t\otimes\bra{t})\right]\\
    &=\sum_{s,t\in\mathcal{O}}\Tr_M\left[(\mathbb{1}_S\otimes s^\dagger\otimes\ket{s})\left(\sum_{i,j,k,l}c^e_{ijkl}\ket{i}\bra{j}\otimes\ket{k}\bra{l}\right)(\mathbb{1}_S\otimes t\otimes\bra{t})\right]\\
    &=\sum_{s,t\in\mathcal{O}}\sum_{i,j}\ket{i}\bra{j}\Tr\left[(\ket{j}\bra{i}\otimes s^\dagger t)  \tilde{\sigma}_{SM}^e\right]\otimes\ket{s}\bra{t}
\end{aligned}
\qquad\forall e\in\{1,\dots,k\}.
\end{equation}
Let's write the objective and the constraints in function of the matrices $\{\Gamma_e\}_e$.

\begin{itemize}
\item \textbf{Objective function.}
\begin{equation}\label{eq:objective}
\begin{aligned}
    \sum_a \sum_{i,j}
    \bra{j}\!\bra{\mathbb{1}}\,
    \Gamma_a\,
    \ket{i}\!|K_M^{i,j|a}\rangle
    &=
    \sum_a \sum_{i,j}
    \Tr\!\left[
        \ket{i}\!\bra{j}
        \otimes
        K_M^{i,j|a}
        \tilde{\sigma}_{SM}^a
    \right] \\
    &=
    \sum_a
    \Tr\!\left[
        \Pi_{SM}^a
        \tilde{\sigma}_{SM}^a
    \right].
\end{aligned}
\end{equation}

\item \textbf{Constraint $\Tr_M[\sum_e \tilde{\sigma}_{SM}^e] = \rho_S$.}
\begin{equation}\label{eq:constr_state}
\begin{aligned}
    \sum_e
    \bra{i}\!\bra{\mathbb{1}}\,
    \Gamma_e\,
    \ket{j}\!\ket{\mathbb{1}}
    &=
    \sum_e
    \Tr\!\left[
        (\ket{j}\!\bra{i} \otimes \mathbb{1})
        \tilde{\sigma}_{SM}^e
    \right] \\
    &=
    \bra{i}\rho_S\ket{j}
\end{aligned}
\qquad
\forall i,j.
\end{equation}

\item \textbf{Constraint $\Tr_M[\Pi_{SM}^a(\mathbb{1}_S \otimes \Tr_S[\sum_e \tilde{\sigma}_{SM}^e])] = M_S^a$.}
\begin{equation}\label{eq:constr_meas}
\begin{aligned}
    \sum_e \sum_l
    \bra{l}\!\bra{\mathbb{1}}\,
    \Gamma_e\,
    \ket{l}\!|K_M^{i,j|a}\rangle
    &=
    \sum_e
    \Tr\!\left[
        (\mathbb{1} \otimes K_M^{i,j|a})
        \tilde{\sigma}_{SM}^e
    \right] \\
    &=
    \bra{i} M_S^a \ket{j}
\end{aligned}
\qquad
\forall a,i,j.
\end{equation}

This follows from the identity
\begin{equation}
\begin{aligned}
    M_S^a
    &=
    \Tr_M\!\left[
        \Pi_{SM}^a
        (\mathbb{1}_S \otimes \Tr_S[\sum_e \tilde{\sigma}_{SM}^e])
    \right] \\
    &=
    \sum_{i,j}
    \ket{i}\!\bra{j}
    \Tr\!\left[
        (\mathbb{1} \otimes K_M^{i,j|a})
        \sum_e \tilde{\sigma}_{SM}^e
    \right].
\end{aligned}
\end{equation}

\item \textbf{Constraint $\Tr[\Pi_{SM}^a \sum_e \tilde{\sigma}_{SM}^e] = \Tr[M_S^a \rho_S]$.}
\begin{equation}\label{eq:constr_statistics}
\begin{aligned}
    \sum_e \sum_{i,j}
    \bra{j}\!\bra{\mathbb{1}}\,
    \Gamma_e\,
    \ket{i}\!|K_M^{i,j|a}\rangle
    &=
    \sum_e
    \Tr\!\left[
        \Pi_{SM}^a
        \tilde{\sigma}_{SM}^e
    \right]
\end{aligned}
\qquad
\forall a.
\end{equation}
\end{itemize}
In addition to the constraints discussed above, further conditions must be imposed on the moment matrices to enforce the algebraic properties of the operators $\{\Pi_{SM}^a\}_a$.
First, since $\{\Pi_{SM}^a\}_a$ are Hermitian, i.e.\ $(\Pi_{SM}^a)^\dagger = \Pi_{SM}^a$ $\forall a$, we require
$\bigl( K_M^{i,j|a} \bigr)^\dagger = K_M^{j,i|a}\;\,\forall\, i,j,a$. In our formulation, this is enforced through
    \begin{equation}
        \bra{m}\!\bra{\mathbb{1}}\,
    \Gamma_e\,
    \ket{n}\!|s\rangle = \bra{m}\!\bra{s^\dagger}\,
    \Gamma_e\,
    \ket{n}\!|\mathbb{1}\rangle  \qquad\;\,\forall s\in\mathcal{O}\quad\forall\, e,m,n.
    \end{equation}
Second, the completeness condition $\sum_a \Pi_{SM}^a=\mathbb{1}_{SM}$ implies
\begin{equation}
    {}_{S}\bra{i}\sum_a\Pi_{SM}^a\ket{j}_S=\sum_{a} K_{M}^{i,j|a}=\delta_{i,j}\mathbb{1}_{M}.
\end{equation}
At the level of the moments, this translates to
\begin{equation}
\begin{aligned}
    \sum_a\bra{m}\!\bra{\mathbb{1}}\,
    \Gamma_e\,
    \ket{n}\!|K_{M}^{i,j|a}\rangle = &\sum_a\Tr[\ket{m}\bra{n}\otimes K_M^{i,j|a}\tilde{\sigma}_{SM}^e]\\
    =&\,\delta_{ij}\Tr[\ket{m}\bra{n}\otimes \mathbb{1}_{M}\tilde{\sigma}_{SM}^e]= \delta_{ij}\bra{m}\!\bra{\mathbb{1}}\,
    \Gamma_e\,
    \ket{n}\!|\mathbb{1}\rangle \qquad\forall\, e,m,n,i,j.
\end{aligned}
\end{equation}
Finally, the projectivity condition $(\Pi_{SM}^a)^2 = \Pi_{SM}^a$ is equivalent to
\begin{equation}
    {}_{S}\bra{i}(\Pi_{SM}^a)^2\ket{j}_S=\sum_lK_M^{i,l|a}K_M^{l,j|a}=K_M^{i,j|a}={}_{S}\bra{i}\Pi_{SM}^a\ket{j}_S.
\end{equation}
In terms of expectation values, we impose
\begin{equation}\label{eq:constr_proj_mom_matr}
\begin{aligned}
    \sum_l \bra{m}\langle (K_M^{l,i|a})^\dagger|\,
    \Gamma_e\,
    \ket{n}\!|K_M^{l,j|b}\rangle&=\Tr\!\left[
        \left(\sum_{l}\ket{m}\bra{n}\otimes K_M^{i,l|a}K_M^{l,j|b} \right)\tilde{\sigma}_{SM}^e
    \right]\\
    &=
    \delta_{ab}\Tr\!\left[
        \left(\ket{m}\bra{n}\otimes K_M^{i,j|a}\right) \tilde{\sigma}_{SM}^e
    \right]\\
    &=\delta_{ab}\bra{m}\!\bra{\mathbb{1}}\,
    \Gamma_e\,
    \ket{n}\!|K_M^{i,j|a}\rangle
    \qquad
    \forall\, a,b,e,m,n,i,j.
\end{aligned}
\end{equation}

\section{Moment matrix separability lemma}\label{App:separability}
\begin{lemma}
    Given a state $\rho_{SM}$ that is separable, then $\Gamma^{T_S}(\rho_{SM})\geq 0.$
\end{lemma}
\begin{proof}
Recall that
\begin{align*}
\Gamma(\rho_{SM})=\sum_{i,j}\ket{i}\bra{j}\otimes \sum_{s,t\in\mathcal{O}}c^{i,j}_{t^\dagger s}(\rho_{SM})\ket{s}\bra{t},
\end{align*}
with
\begin{align*}
    c^{i,j}_{s}(\rho_{SM})=\Tr\left[(\ket{j}\bra{i}\otimes s)  \rho_{SM}\right].
\end{align*}
Next, from the fact that
\begin{align*}
    \Tr[(\ket{i}\bra{j}_S\otimes s)\sigma_{SM}^{T_S}]=\Tr[(\ket{j}\bra{i}_S\otimes s)\sigma_{SM}]
\end{align*}
we have that
\begin{align*}
\Gamma(\rho_{SM})^{T_S}&=\sum_{i,j}\ket{j}\bra{i}\otimes \sum_{s,t\in\mathcal{O{}}}c^{ij}_{t^\dagger s}(\rho_{SM})\ket{s}\bra{t}\\
&=\sum_{i,j}\ket{j}\bra{i}\otimes \sum_{s,t\in\mathcal{O}}c^{ji}_{t^\dagger s}(\rho_{SM}^{T_S})\ket{s}\bra{t}\\
&=\Gamma(\rho_{SM}^{T_S}).
\end{align*}
The lemma then follows from the facts that $\rho_{SM}^{T_S}\geq 0$ for a separable $\rho_{SM}$ and $\rho\geq 0\implies \Gamma(\rho)\geq 0$.
\end{proof}

\section{Sufficiency of the SDP relaxation}\label{app:proof_SDP_formulation}
This section is devoted to the proof of Theorem \ref{th:SDP_formulation} from the main text. In fact, we prove a somewhat stronger result, Proposition \ref{prop:charac} below, from which the said theorem follows as a corollary.

\medskip

In the following, we use the symbol $(K^{i,j|a}_{M})^\dagger$ as a shorthand for $K^{ji|a}_{M}$. The goal of this section is to prove the following result.
\begin{prop}
\label{prop:charac}
For $e\in\{1,...,k\}$, let $\{c^{ij}_e(s):i,j=1,...,d,s\in {\cal O}^2\}\subset \CC$ be such that 
\begin{align}
&\sum_{a=1}^{k}c_e^{mn}(K^{i,j|a}_{M})=\delta_{ij}c^{mn}_e(\mathbb{1}),\label{cond_compl_out}\\
&\sum_{l=1}^{d}c^{mn}_e((K^{l,i|a}_M)^\dagger K^{l,j|b}_M)=\delta_{ab}c^{mn}_e(K^{i,j|b}_M)\label{cond_orth_out}.
\end{align}
and the matrix
\begin{equation}
\Gamma_e:=\sum_{s,t\in {\cal O}}\sum_{i,j=1}^dc^{ij}_e(ts^\dagger)\ket{i}\bra{j}\otimes \ket{s}\bra{t}
\end{equation}
be positive semidefinite. Then, there exist a Hilbert space $\tilde{{\cal H}}$, non-normalized quantum states states $\{\tilde{\sigma}^e_{SM}\}_e\subset B(\CC^d)\otimes B(\tilde{{\cal H}})$ and operators $\{\tilde{K}^{ij|a}_M\}\subset B(\tilde{{\cal H}})$ such that
\begin{equation}
\tilde{\Pi}^a_{SM}:=\sum_{i,j=1}^d \ket{i}\bra{j}\otimes\tilde{K}^{i,j|a}_M,\quad \forall a\in\{1,...,k\}
\end{equation}
defines a projective measurement and
\begin{equation}
\tr(\ket{j}\bra{i}\otimes \tilde{t}\tilde{s}^\dagger\tilde{\sigma}^e_{SM})=c_e^{ij}(ts^\dagger),
\label{rel_with_cs}
\end{equation}
for $s,t\in {\cal O}$.
 
\end{prop}

To prove the proposition, we will make use of the following lemma.
\begin{lemma}
\label{lemma:charac}
Let $\{c^{ij}(s):i,j=1,...,d,s\in {\cal O}^2\}\subset\CC$ satisfy the constraints
\begin{align}
&\sum_{a=1}^kc^{mn}(K^{i,j|a}_{M})=\delta_{ij}c^{mn}(\mathbb{1}),\label{cond_compl}\\
&\sum_{l=1}^{d}c^{mn}((K^{l,i|a}_M)^\dagger K^{l,j|b}_M)=\delta_{ab}c^{mn}(K^{i,j|b}_M)\label{cond_orth}.
\end{align}
Let the matrix
\begin{equation}
\Gamma=\sum_{s,t\in {\cal O}}\sum_{i,j=1}^dc^{ij}(ts^\dagger)\ket{i}\bra{j}\otimes \ket{s}\bra{t}
\end{equation}
be positive semidefinite. Then, there exist a Hilbert space $\tilde{{\cal H}}$ of dimension $d|{\cal O}|$, a non-normalized quantum state $\tilde{\sigma}\in B(\CC^d)\otimes B(\tilde{{\cal H}})$ and operators $\{\tilde{K}^{i,j|a}_M\}_{i,j,a}\subset B(\tilde{{\cal H}})$ such that
\begin{equation}
\tilde{\Pi}^a_{SM}:=\sum_{i,j=1}^d \ket{i}\bra{j}\otimes\tilde{K}^{i,j|a}_M,\quad \forall a\in \{1,...,k\}
\label{exp_projs}
\end{equation}
defines a projective measurement and
\begin{equation}
\tr(\ket{j}\bra{i}\otimes \tilde{t}\tilde{s}^\dagger\tilde{\sigma}_{SM})=c^{ij}(ts^\dagger),
\label{phys_real}
\end{equation}
for $s,t\in {\cal O}$.

\end{lemma}
The intuition of the proof of this lemma is that $\sigma_{SM}$ can be chosen to be of the form $\sigma_{SM}=\ket{\tilde{\psi}}\bra{\tilde{\psi}}$, with
\begin{equation}
\ket{\psi}=\sum_{m=1}^d\ket{m}\ket{\tilde{\psi}_m}.
\end{equation}
If there exist projector operators of the form (\ref{exp_projs}) and relation (\ref{phys_real}) is satisfied, then we have that:
\begin{equation}
c^{mn}(ts^\dagger)=\bra{\tilde{\psi}_n}\tilde{t}\tilde{s}^\dagger\ket{\tilde{\psi}_m}.
\label{bare_overlaps}
\end{equation}
In the proof, we will use a guess on $\{\tilde{s}\ket{\psi_m}:s\in{\cal O}\}$ to construct the projector operators $\{\tilde{\Pi}^a_{SM}\}_a$.
\begin{proof}
Since $\Gamma$ is positive semidefinite, it admits a Gram decomposition, i.e., there exist vectors $\{\ket{m, s}:m=1,...,d,s\in {\cal O}\}\subset \CC^{d|{\cal O}|}=:\tilde{{\cal H}}$ such that
\begin{equation}
c^{mn}(ts^\dagger)=\braket{m,s}{n,t}.
\end{equation}
Comparing this expression with Eq. (\ref{bare_overlaps}) suggests one to define the vectors
\begin{equation}
\ket{\psi_m, s^\dagger}:=\overline{\ket{m,s}},
\end{equation}
which therefore satisfy
\begin{equation}
c^{mn}(t^\dagger s)=\braket{\psi_n, t}{\psi_m,s}.
\label{overlaps_good_vectors}
\end{equation}
Intuitively, $\ket{\psi_m,s}$ would represent the vector $\tilde{s}\ket{\tilde{\psi}_m}$ alluded to above; accordingly, the vector
\begin{equation}
\ket{\Phi^j_{m,a}}:=\sum_i\ket{i}\ket{\psi_m,K_M^{i,j|a}}
\end{equation}
would correspond to $\tilde{\Pi}^a_{SM}\ket{j}\ket{\tilde{\psi}_m}$. In this regard, note that
\begin{align}
&\braket{\Phi^i_{m,a}}{\Phi^j_{n,b}}=\sum_{l}\braket{\psi_m, K^{l,i|a}_M}{\psi_n, K^{l,j|b}_M}\nonumber\\
&=\sum_{l}c^{nm}((K^{l,i|a}_M)^\dagger K^{l,j|b}_M)\nonumber\\
&=\delta_{ab}c^{nm}(K^{i,j|b}_M);
\end{align}
in particular, $\ket{\Phi_{m,a}^i}\perp \ket{\Phi^j_{n,b}}$ for $a\not=b$. For $a=1,...,k$, the subspaces
\begin{equation}
H_a:=\mbox{span}\{\ket{\Phi^i_{m,a}}:m,i\}    
\end{equation}
are therefore pair-wise orthogonal; this suggests defining the projectors
\begin{equation}
\tilde{\Pi}_{SM}^a:=\mbox{proj}\left(H_a\right),\quad \forall a\in \{1,...,k\}.
\end{equation}
The orthogonality relations between the subspaces $\{H_a\}_a$ guarantee that $\{\tilde{\Pi}_{SM}^a\}_a$ are either a complete or an undercomplete set of projectors. In the latter case, we redefine $\tilde{\Pi}^1_{SM}:=\mathbb{1}-\sum_{a\not=1}\tilde{\Pi}_{SM}^a$.

Now, consider the vector
\begin{align}
&\ket{i}\ket{\psi_m,\mathbb{1}}-\sum_a\ket{\Phi^i_{m,a}}.
\end{align}
This vector is null. Indeed, its norm squared is
\begin{align}
&\braket{\psi_m,\mathbb{1}}{\psi_m,\mathbb{1}}+\sum_a\braket{\Phi^i_{m,a}}{\Phi^i_{m,a}}-\sum_a\bra{i}\bra{\psi_m,\mathbb{1}}\ket{\Phi^i_{m,a}}-\sum_a\bra{\Phi^i_{m,a}}\ket{i}\ket{\psi_m,\mathbb{1}}\nonumber\\
&=\braket{\psi_m,\mathbb{1}}{\psi_m,\mathbb{1}}+\sum_a\sum_j\braket{\psi_m,K_M^{j,i|a}}{\psi_m,K_M^{j,i|a}}-
\sum_a\braket{\psi_m,\mathbb{1}}{\psi_m,K_M^{i,i|a}}-\sum_a\braket{\psi_m,K_M^{i,i|a}}{\psi_m,\mathbb{1}}\nonumber\\
&=c^{mm}(\mathbb{1})+\sum_a\sum_jc^{mm}((K_M^{j,i|a})^\dagger K_M^{j,i|a})-\sum_ac^{mm}(K_M^{i,i|a})-\sum_ac^{mm}(K_M^{i,i|a})\nonumber\\
&=c^{mm}(\mathbb{1})+\sum_ac^{mm}(K_M^{i,i|a})-2\sum_ac^{mm}(K_M^{i,i|a})\nonumber\\
&=c^{mm}(\mathbb{1})-\sum_ac^{mm}(K_M^{i,i|a})=0,
\end{align}
where we made use of relations (\ref{cond_compl}), (\ref{cond_orth}).
Thus, for all $m,i$, it holds that
\begin{equation}
\ket{i}\ket{\psi_m,\mathbb{1}}=\sum_{a}\ket{\Phi^i_{m,a}}.  \end{equation}
By the orthogonality relations $\{H_a\perp H_b:a\not=b\}$, this implies that
\begin{align}
&\tilde{\Pi}^a_{SM}\ket{i}\ket{\psi_m,\mathbb{1}}=\tilde{\Pi}_{SM}^a\sum_b\ket{\Phi^i_{m,b}}\nonumber\\
&=\ket{\Phi^i_{m,a}}.
\label{action_proj}
\end{align}
Let $\{\tilde{K}^{ij|a}_{SM}\}_{a,i,j}\subset B(\tilde{{\cal H}})$ be such that
\begin{equation}
\tilde{\Pi}_{SM}^a=\sum_{i,j}\ket{i}\bra{j}\otimes \tilde{K}^{i,j|a}_M,
\end{equation}
and define
\begin{equation}
\ket{\tilde{\psi}}:=\sum_m\ket{m}\ket{\psi_m,\mathbb{1}}.
\end{equation}
Then we have that
\begin{align}
&\bra{m}\otimes \tilde{K}^{j,k|a}_M\ket{\tilde{\psi}}\nonumber\\
&=(\bra{j}\otimes\mathbb{1}) \tilde{\Pi}^a_{SM}(\ket{k}\bra{m}\otimes \mathbb{1})\ket{\tilde{\psi}}\nonumber\\
&=(\bra{j}\otimes\mathbb{1}) \tilde{\Pi}^a_{SM}\ket{k}\ket{\psi_m,\mathbb{1}}\nonumber\\
&=(\bra{j}\otimes\mathbb{1}) \ket{\Phi^k_{m,a}}\nonumber\\
&=\ket{\psi_m,K^{j,k|a}_M},
\label{action_Ks}
\end{align}
where we invoked eq. (\ref{action_proj}) and the definition of $\ket{\Phi^i_{a,m}}$ in the last two lines.

Finally, define
\begin{equation}
\tilde{\sigma}_{SM}:=\proj{\tilde{\psi}}.
\end{equation}
By eqs. (\ref{overlaps_good_vectors}) and (\ref{action_Ks}), for all $s,t\in {\cal O}$ it holds that
\begin{align}
&\tr(\ket{n}\bra{m}\otimes \tilde{t}\tilde{s}^\dagger\tilde{\sigma}_{SM})\nonumber\\
&=\bra{\tilde{\psi}}(\ket{n}\otimes\tilde{t})(\bra{m}\otimes\tilde{s}^\dagger)\ket{\tilde{\psi}}\nonumber\\
&=\braket{\psi_n,\tilde{t}^\dagger}{\psi_m,\tilde{s}^\dagger}=c^{mn}(ts^\dagger).
\end{align}
\end{proof}
Now we are ready to prove the proposition.
\begin{proof}[Proof of Proposition \ref{prop:charac}] 
By Lemma \ref{lemma:charac}, there exists a Hilbert space $\hat{{\cal H}}$ of dimension $d|{\cal O}|$ such that, for $e=1,...,k$, there is a set of operators $\{\hat{s}(e):s\in {\cal O}^2\}\subset B(\hat{{\cal H}})$ and a non-normalized state $\hat{\sigma}^e_{SM'}\subset B(\CC^d\otimes \hat{{\cal H}})$ with the properties:
\begin{enumerate}
    \item $\tr(\ket{j}\bra{i}\otimes \hat{t}(e)\hat{s}(e)^\dagger\hat{\sigma}_e)=c^{ij}_e(ts^\dagger)\;\,\forall s,t\in {\cal O}$.
    \item The operators
    \begin{equation}
    \tilde{\Pi}_{SM}^a:=\sum_{i,j=1}^d\ket{i}\bra{j}\otimes\tilde{K}^{i,j|a}_M(e),\quad \forall a\in \{1,...,k\}
    \end{equation}
    define a complete projective measurement.
\end{enumerate}

We next integrate all the operators above in an extended Hilbert space $\tilde{{\cal H}}=\hat{{\cal H}}\otimes \CC^{k}$, through the relation
\begin{equation}
\tilde{t}:=\sum_{e=1}^{k}\hat{t}(e)\otimes\proj{e}.
\end{equation}
Similarly, we define the states
\begin{equation}
\tilde{\sigma}^e_{SM}:=\hat{\sigma}^e_{SM'}\otimes \proj{e}.
\end{equation}
It can then be verified that the operators
\begin{equation}
\tilde{\Pi}^a_{SM}=\sum_{i,j}\ket{i}\bra{j}\otimes \tilde{K}^{i,j|a}_M=\sum_e\hat{\Pi}_{SM'}^a(e)\otimes \proj{e}
\end{equation}
form a complete system of projectors. Moreover, Eq. (\ref{rel_with_cs}) holds, since, for all $t,s\in {\cal O}$,
\begin{align}
&\tr(\ket{j}\bra{i}\otimes \tilde{t}\tilde{s}^\dagger\tilde{\sigma}^e_{SM})\nonumber\\
&=\tr(\ket{j}\bra{i}\otimes \hat{t}(e)\hat{s}^\dagger(e)\hat{\sigma}_{SM'}^e)=c_e^{ij}(ts^\dagger).
\end{align}

\end{proof}

\begin{proof}[Proof of Theorem \ref{th:SDP_formulation}] 
Let $\{\Gamma_e\}_e$, with 
\begin{equation}
\Gamma_e:=\sum_{s,t\in {\cal O}}\sum_{i,j=1}^dc^{ij}_e(ts^\dagger)\ket{i}\bra{j}\otimes \ket{s}\bra{t}
\end{equation}
be a feasible solution to Eq. \eqref{eq:pguess_matrices} achieving an objective value $c$.
By identifying the coefficients $c_e^{i,j}(ts^\dagger)=\bra{i}\bra{s}\Gamma_e\ket{j}\ket{t}$, we can see that the constraints in the SDP formulation directly enforce the hypotheses of Proposition~\ref{prop:charac}. Specifically
\begin{enumerate}
    \item The completeness constraint $\sum_a\bra{m}\!\bra{\mathbb{1}}\,
    \Gamma_e\,
    \ket{n}\!|K_{M}^{i,j|a}\rangle = \delta_{ij}\bra{m}\!\bra{\mathbb{1}}\,
    \Gamma_e\,
    \ket{n}\!|\mathbb{1}\rangle$ translates to~\eqref{cond_compl_out}.
    \item The projectivity constraint $\sum_l \bra{m}\langle (K_M^{l,i|a})^\dagger|\,
    \Gamma_e\,
    \ket{n}\!|K_M^{l,j|b}\rangle =\delta_{ab}\bra{m}\!\bra{\mathbb{1}}\,
    \Gamma_e\,
    \ket{n}\!|K_M^{i,j|a}\rangle$ translates to~\eqref{cond_orth_out}.
    \item The constraint $\Gamma_e \succeq 0$ ensures the matrices are positive semidefinite.
\end{enumerate}
    Since all conditions of Proposition~\ref{prop:charac} are met, there exists a Hilbert space $\tilde{\mathcal{H}}$, subnormalized states $\{\tilde{\sigma}_{SM}^{e}\}_e$, and operators $\{\tilde{K}_M^{i,j|a}\}_a$ forming a valid projective measurement $\tilde{\Pi}^a_{SM}=\sum_{i,j=1}^d\ket{i}\bra{j}\otimes \tilde{K}_M^{i,j|a}$ such that
    \begin{equation}
        \Tr[(\ket{j}\bra{i}\otimes\tilde{t}\tilde{s}^\dagger)\tilde{\sigma}^e_{SM}]=c^{ij}_e(ts^\dagger).
    \end{equation}
    for $s,t\in\mathcal{O}$.
    
    It directly follows from Equations~\eqref{eq:objective}-\eqref{eq:constr_statistics}, that these constructed states and measurements form a feasible solution to the original non-commutative polynomial optimization problem in Eq.~\eqref{eq:pguess_SM} and that achieves the same optimal value $c$. 
    \end{proof} 
\end{document}